\documentclass[submission,copyright,creativecommons]{eptcs}
\providecommand{\event}{ICE 2021} % Name of the event you are submitting to
\usepackage{breakurl}             % Not needed if you use pdflatex only.
\usepackage{underscore}           % Only needed if you use pdflatex.
\usepackage[UKenglish]{babel}

\usepackage{amssymb}
\usepackage{amsfonts}

\usepackage{tikz}
\usetikzlibrary{automata, positioning, arrows, petri, shapes.geometric}

\tikzset{
	>=stealth, % makes the arrow heads bold
	node distance=2cm, % specifies the minimum distance between two nodes. Change if necessary.
	every node/.style={scale=0.7},
	every state/.style={thick, fill=gray!10}, % sets the properties for each ’state’ node
	initial text=$ $, % sets the text that appears on the start arrow
	elliptic state/.style={draw,ellipse}
}

\usepackage{calc}
\usepackage{xspace}

\usepackage[utf8]{inputenc}
\usepackage{float}
\usepackage{tikz} 
\usetikzlibrary{shapes}
\usetikzlibrary{automata}
\usetikzlibrary{calc}
\usepackage{bussproofs}
\usepackage{mathtools}
\usepackage{mathpartir}
\usepackage{todonotes}
\usepackage{stmaryrd}

\usepackage{algorithm}
\usepackage[noend]{algpseudocode}

\usepackage{stackengine}
\stackMath

\usepackage{macrosscm}
\usepackage{macrosmsc}
\usepackage{macros}

\usepackage{csquotes}

\usepackage{subfigure}

\usepackage{amsthm}
\newtheorem{theorem}{Theorem}
\newtheorem{lemma}[theorem]{Lemma}
\newtheorem{definition}[theorem]{Definition}

\theoremstyle{remark}
\newtheorem{example}{Example}

\usepackage{graphicx}

\title{Towards Generalised Half-Duplex Systems\footnote{This work has been supported by the French government, through the EUR DS4H Investments in the Future project managed by the National Research Agency (ANR) with the reference number ANR-17-EURE- 0004.}}
\def\titlerunning{Towards Generalised Half-Duplex Systems}

\author{Cinzia {Di Giusto}  \qquad Loïc Germerie Guizouarn \qquad Etienne Lozes
\institute{Université Nice Côte d'Azur\\CNRS, I3S\\Sophia Antipolis, France}
\email{\{cinzia.di-giusto,loic.germerie-guizouarn,etienne.lozes\}@univ-cotedazur.fr}
}
\def\authorrunning{C. Di Giusto et al.}
\begin{document}

\maketitle              % typeset the header of the contribution

\begin{abstract}
	% !TEX root = ice.tex

FIFO automata are finite state machines communicating through FIFO queues.
They can be used for instance to model distributed protocols.
Due to the unboundedness of the FIFO queues, 
several verification problems are undecidable for these systems.
In order to model-check such systems, one may look for
decidable subclasses of FIFO systems.
Binary half-duplex systems are systems of two FIFO automata exchanging over a half-duplex channel.
They were studied by Cece and Finkel who established the
decidability in polynomial time of several properties.
These authors also identified some problems in generalising half-duplex systems to multi-party 
communications.
We introduce greedy systems, as a candidate to generalise binary half-duplex systems.
We show that greedy systems retain the same good properties as binary half-duplex systems, and 
that, in the setting of mailbox communications, greedy systems are quite closely related to a multiparty 
generalisation of half-duplex systems.
\end{abstract}

\section{Introduction}\label{sec:introduction}
% !TEX root = ice.tex

FIFO automata, also known as asynchronous communicating automata (i.e., finite state automata that  exchange messages via FIFO queues) are an interesting formalism for modeling distributed protocols. In their most general formulation, these automata are Turing powerful, and in order to be able to model check them it is necessary to reduce their expressiveness. 

% To this aim, many classes have been introduced in the literature depending either on restrictions on the topology, or on relaxing the FIFO discipline, or on combinations of the two~\cite{DBLP:conf/concur/ChambartS08,DBLP:conf/concur/ClementeHS14,DBLP:journals/iandc/CeceFI96,DBLP:journals/iandc/AbdullaJ96}. Other restrictions have been considered, in particular about the size of the FIFO queues, see below for a detailed discussion.
% Without such restrictions, however, even the simplest setting of two machines communicating with two queues is undecidable.

Binary \halfduplex systems, introduced by Cece and Finkel~\cite{DBLP:journals/iandc/CeceF05}, are systems with two participants and a bidirectional channel formed of two FIFO queues, such that communication happens only in one direction at a time. The stereotypical \halfduplex device is the walkie-talkie (or the CB). In several applications, in particular when
FIFO buffers are bounded and sends may be blocking,
\halfduplex communications are considered a good practice to avoid send-send deadlocks. Language support for enforcing this discipline of communication includes, for instance, binary session types~\cite{DBLP:conf/concur/Honda93,DBLP:conf/esop/HondaVK98} or Sing\# channel contracts~\cite{DBLP:conf/eurosys/FahndrichAHHHLL06,DBLP:conf/wsfm/LozesV11}.

In~\cite{DBLP:journals/iandc/CeceF05}, Cece and Finkel
show that (1) whether a system is half-duplex is decidable in 
polynomial time, (2) the set of reachable configurations is 
regular, and (3) properties like progress and 
boundedness are decidable in polynomial time. Cece and Finkel also present two possible 
notions \textquote{multiparty half-duplex} systems generalizing their 
class to systems of any number of machines for 
p2p communications (one FIFO queue per pair of machine).

The first generalisation involves assuming that
at most one queue over all queues is 
non-empty at any time. This generalisation preserves 
decidability but is very restrictive.
The second generalisation restricts the communications between each pair of participants to \halfduplex communications, that is only one buffer per bidirectional channel can be used simultaneously. This generalisation however does not preserve decidability: the systems captured by this definition form a Turing powerful class. In fact, with just three machines it is possible to mimic the tape of a Turing machine.

It could be believed that these results end the discussion
about multi-party half-duplex systems. In this work, we
claim conversely that there is another natural and 
relevant notion of multi-party half-duplex communications
that allows us to generalise the results of Cece and Finkel.
We introduce \emph{greedy systems}, which are systems for which
all executions can be rescheduled in such a way that
all receptions are immediately preceded by their corresponding
send. This notion is quite natural, and
closely related to other notions like
synchronisability~\cite{DBLP:journals/tcs/BasuB16},
$1$-synchronous systems~\cite{DBLP:conf/cav/BouajjaniEJQ18},
or existentially $1$-bounded 
system~\cite{DBLP:journals/fuin/GenestKM07} (see Section~\ref{sec:conclusion} for a detailed discussion).
%%% Relire cette partie %%%

%%%%%%%%%%%%%%%%

In this work, we establish the following results:
\begin{enumerate}
    \item whether a system is greedy is decidable in polynomial 
    time (when the number of processes is fixed);
    \item for greedy systems, all regular safety properties,
    which includes reachability, absence of unspecified 
    receptions, progress, 
    and boundedness are decidable in polynomial time.
    \item we generalize binary half-duplex systems to multiparty mailbox half-duplex systems and
    we show that (1) mailbox half-duplex systems are greedy, and (2) greedy systems
    without orphan messages, at least in the binary case, are half-duplex.
    
\end{enumerate}
The first result follows from techniques developed
by Bouajjani \emph{et al}~\cite{DBLP:conf/cav/BouajjaniEJQ18} for 
$k$-synchronous systems. The main challenge here is that
we address a more general model of communicating systems
that encompasses both mailbox and p2p communications, but also
allows any form of sharing of buffers among processes.
The second result is based on an approach that, to
the best of our knowledge, is new, although it borrows 
from some
general principles from regular model-checking.
The challenge is that, unlike for binary half-duplex 
systems, the reachability set of greedy systems is not regular,
which complicates how automata-based techniques can be used
to solve regular safety.
The third contribution aims at answering, although
incompletely, the question we would like to address with
this work: what is a relevant notion of multi-party
half-duplex systems?

\paragraph{Outline.}
The paper is organised as follows: 
Section~\ref{sec:preliminaries} introduces communicating 
automata and systems.
Section~\ref{sec:greedy-def} defines greedy systems and 
establishes the decidability of the greediness of a system.
Section~\ref{sec:model-checking} discusses regular
safety for greedy systems.
Section~\ref{sec:mailbox-hd} compares greedy systems
and half-duplex systems, first in the binary setting,
then in the multi-party setting, by introducing
the notion of mailbox half-duplex systems.
% it also studies
%how this definition does not adapt well to p2p systems. 
Finally, Section~\ref{sec:conclusion} concludes with some final remarks and discusses related works.

\section{Preliminaries} \label{sec:preliminaries}
% !TEX root = ice.tex

For a finite set $S$, $S^*$ denotes the set of finite words over $S$,
$w_1\cdot w_2$ denotes the concatenation of two words, $|w|$ denotes the length of $w$, and $\epsilon$ denotes the empty word. We assume some familiarity with non-deterministic
finite state automata,
and we write $\languageof{\A}$ for the language accepted
by the automaton $\A$. For two sets $S$ and $I$,
we write $\vect{b}$ (in bold) for an element of $S^I$,
and $b_i$ for the $i$-th component of $\vect{b}$, so that
$\vect{b}=(b_i)_{i\in I}$.

A \emph{FIFO automaton} is basically a finite state machine equipped with 
FIFO queues where transitions are labelled with either queuing or dequeuing actions. More precisely:
\begin{definition}[FIFO automaton]
  A \emph{FIFO automaton} is a tuple\footnote{
Note that FIFO automata do not have accepting states, therefore they are
not a special case of non-deterministic finite state automaton, and
there is not such a thing as "the language of a FIFO automaton".
  }  
  $\A=(L,\paylodSet,I,\actionSet,\delta,l_0)$
  where
(1) $L$ is a finite set of \emph{control states},
(2) $\paylodSet$ is a finite set of \emph{messages},
(3) $I$ is a finite set of \emph{buffer identifiers},
(4) $\actionSet\subseteq I\times\{!,?\}\times\paylodSet$ is a finite set of \emph{actions},
(5) $\delta\subseteq L\times \actionSet \times L$ is the \emph{transition relation}, and
(6) $l_0\in L$ is the \emph{initial control state}.
The \emph{size} $|\A|$ of $\A$ is $|L|+|\delta|$.
\end{definition}

Actions $a=(i,\dag,\amessage)$,  for $\dag \in \{!,?\}$ are also 
denoted by $i\dag\amessage$. 
Given a FIFO automaton $\A=(L,\paylodSet,I,\delta, l_0)$, a \emph{configuration} of $\A$ is a tuple $\aconf=(l,\vect{b})\in L\times \paylodSet^I$. The initial configuration 
is $\initconf=(l_0,\emptybuffers)$ where for all $i\in I$,
$\emptybufferat{i}=\epsilon$.
A \emph{step} is a tuple 
$(\aconf,a,\aconf')$ (often written
$\transition{\aconf}{a}{\aconf'}{\A}$)
where $\aconf=(l,\vect{b})$ and $\aconf'=(l',\vect{b}')$ are configurations and $a$ is an action,
such that the following holds:
\begin{itemize}
  \item $(l,a,l')\in\delta$,
  \item if $a=i!\amessage$, then $b'_i=b_i\cdot \amessage$
  and $b_j'=b_j$ for all $j\in I\setminus\{i\}$.
  \item if $a=i?\amessage$, then $b_i=\amessage\cdot b_i'$
  and $b_j'=b_j$ for all $j\in I\setminus\{i\}$.
\end{itemize}

Next, we define systems of FIFO automata. We pick
a very general definition, where each FIFO queue
might be queued (resp. dequeued) by more than one automaton,
and where an automaton might \textquote{send a message to itself}.
Most of the theory can be done in this general setting without
extra cost, but we merely have in mind either mailbox
systems or p2p systems (defined below).

A FIFO system, later called simply a \emph{system}, 
is a family $\system=(\A_p)_{p\in\procSet}$ of FIFO automata 
such that all FIFO automata have 
disjoint sets of actions: for all $p\neq q\in\procSet$, $\actionSet_p\cap\actionSet_q=\emptyset$. 
Each $p\in\procSet$ is referred to as a \emph{process}. The condition on the disjointness 
of the sets of actions helps to identify the process that is responsible for a given action: for
an action $a$, we write $\processof{a}$ to denote
the unique $p$ (when it exists) such that $a\in\actionSet_p$.

Let us now define mailbox and p2p systems.
Informally, a p2p system is a system in which each pair of processes has a dedicated buffer for exchanging messages.
Instead, for  mailbox communication,  each process receives messages from all other processes in a single buffer.
Let us now formally define these notions.
To this aim, it will be useful to identify what are the buffers an automaton queues in (resp. dequeues from). 
Hence, for a given FIFO automaton 
$\A_p=(L_p,\paylodSet_p,I_p,\actionSet_p,\delta_p,l_{0,p})$,
and for $\dag\in\{!,?\}$,
we write $I_p^\dag$ for the set of buffer identifiers $i$
such that there exists $\amessage\in\paylodSet_p$
such that $i\dag\amessage\in\actionSet_p$. 
A system $\system=(\A_p)_{p\in\procSet}$ is \emph{p2p} if
for all $p\in\procSet$, 
$I_p \subseteq \procSet^2$, $I^!_p=\{p\}\times(\procSet\setminus\{p\})$, and $I^?_p=(\procSet\setminus\{p\})\times\{p\}$.
A system $\system=(\A_p)_{p\in\procSet}$ 
is \emph{a mailbox system} if for all $p\in\procSet$,
$I_p\subseteq \procSet$, 
$I_p^{!}=\procSet\setminus\{p\}$ and $I_p^{?}=\{p\}$.  
Thus in a p2p system  with $n$ processes, there are at most $n (n-1)$ buffers, and in a mailbox system with $n$ processes
there are at most $n$ buffers.
A \emph{binary system} is a system $\system=(\A_p)_{p\in\procSet}$ such that $\procSet=\{1,2\}$ and for all $p\in\procSet$,
$I_p^{!}=\{3-p\}=I_{3-p}^{?}$; note that a binary
system is both p2p and mailbox.
We sometimes use a more handy notation
for actions of a mailbox (resp. p2p) system:
if $\processof{i!\amessage}=p$ and
$\processof{i?\amessage}=q$, we sometimes write
$\send{p}{q}{\amessage}$ instead of $i!\amessage$
and $\rec{p}{q}{\amessage}$ instead of $i?\amessage$.

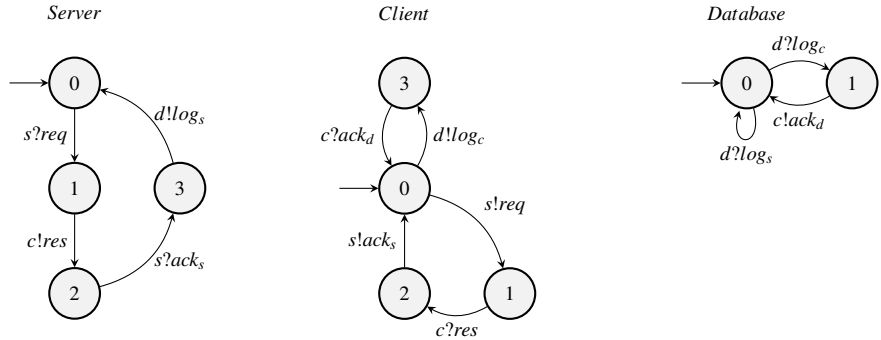
\begin{figure}
\centering
\begin{tikzpicture}
	\node (s) {$Server$};
	\node[state, initial, below=0.4cm of s] (s0) {$0$};
	\node[state, below of=s0] (s1) {$1$};
	\node[state, below of=s1] (s2) {$2$};
	\node[state, right of=s1] (s3) {$3$};
	
	\node[right=3.5cm of s] (c) {$Client$};
	\node[state, below=0.4cm of c] (c3) {$3$};
	\node[state, initial, below of=c3] (c0) {$0$};
	\node[state, below of=c0] (c2) {$2$};	
	\node[state, right of=c2] (c1) {$1$};

	\node[right=3.5cm of c] (d) {$Database$};
	\node[state, initial, below=0.4cm of d] (d0) {0};
	\node[state, right of=d0] (d1) {$1$};
	
	\tikzset{->,}
	\draw
	(s0) edge[left] node{$\crec{c}{s}{req}$} (s1)
	(s1) edge[left] node{$\csend{s}{c}{res}$} (s2)
	(s2) edge[right, bend right] node{$\crec{c}{s}{ack_s}$} (s3)
	(s3) edge[right, bend right] node{$\csend{s}{d}{log_s}$} (s0)
	
	(c0) edge[above right, bend left] node{$\csend{c}{s}{req}$} (c1)
	(c1) edge[below, bend left] node{$\crec{s}{c}{res}$} (c2)
	(c2) edge[left] node{$\csend{c}{s}{ack_s}$} (c0)
	(c0) edge[right, bend right] node{$\csend{c}{d}{log_c}$} (c3)
	(c3) edge[left, bend right] node{$\crec{d}{c}{ack_d}$} (c0)
	
	(d0) edge[loop below, below] node{$\crec{s}{d}{log_s}$} (d0)
	(d0) edge[above, bend left] node{$\crec{c}{d}{log_c}$} (d1)
	(d1) edge[below, bend left] node{$\csend{d}{c}{ack_d}$} (d0);
	
\end{tikzpicture}

\caption{Client/Server/Database protocol}\label{fig:ex:csd}
\end{figure}

\begin{example}[FIFO Automata]
Figure~\ref{fig:ex:csd} shows a graphical representation of 
a FIFO system, borrowed from~\cite{DBLP:journals/fmsd/AkrounS18}.
This system represents a protocol  between a client, a server 
and a database logging requests from the client and the server.
In this protocol, a client can log something on the database or  
send requests to the server, when those requests are satisfied 
the server logs them in a database. 
%The database can also notify the server that it has reached its maximum capacity. 
Each automaton is equipped with a buffer in which it receives 
messages from all other participants: this system is 
an example of a mailbox system. 
To improve readability of the graphical representation, we 
refer to the buffers with the initial of the automaton to 
which they are associated. For example, $s$ is the buffer into 
which the server can receive messages.
This simple system  will be used as a running example throughout the paper. \qed
\end{example}

The size $|\system|$ of $\system$ is the sum of
the $|\A_p|$. 
Note that $|\procSet|$ and $|I|$ are independent from the size of 
$\system$. In particular,
when we say that an algorithm is in polynomial time,
we mean in time $\bigO{|\system|^k}$ for some $k$
that may depend on $|\procSet|$ and $|I|$.

The FIFO automaton $\productof{\system}$ associated with 
$\system$ is the standard asynchronous product automaton: 
$\productof{\system}=(\Pi_{p\in\procSet}L_p, \bigcup_{p\in\procSet}\paylodSet_p,\bigcup_{p\in\procSet}I_p,\delta,\vect{l}_0)$ 
where $\vect{l}_0=(l_{0,p})_{p\in\procSet}$
and $\delta$ is the set of triples $(\vect{l},a,\vect{l}')$
for which there exists $p\in\procSet$ such that $(l_p,a,l_p')\in\delta_p$ and $l_q=l_q'$ for all $q\in\procSet\setminus\{p\}$.
We often identify $\system$ and $\productof{\system}$
and write for instance $\transition{}{a}{}{\system}$ instead of 
$\transition{}{a}{}{\productof{\system}}$. Similarly,
we say that $\aconf=(\vect{l},\vect{b})$
is a configuration of $\system$ while we mean 
that $\aconf$ is a configuration of $\productof{\system}$.
An  execution $e=a_1\cdots a_n\in\actionSet^*$ is a sequence of actions.
As usual $\xRightarrow{e}$ stands for $ \xrightarrow{a_1} \cdots \xrightarrow{a_n}$.
We write $\executionsof{\system}$ for $\{e\in\actionSet^*\mid \Transition{\aconf_0}{e}{\aconf}{\system}\mbox{ for some }\aconf \}$.

Next, we introduce the definition of reachable configuration:

\begin{definition}[Reachable configuration]\label{reachconf}
Let  $\system$ be a system. A configuration $\aconf$ is reachable if there exists $e\in\actionSet^*$ such that   $\Transition{\aconf_0}{e}{\aconf}{\system}$. The set of all reachable configurations of $\system$ is denoted  $\erress{\system}$.
\end{definition}

Given an execution $e=a_1\cdots a_n$, we say that
$\{j,j'\}\subseteq\{1,\cdots,n\}^2$ is a \emph{matching pair} if
there exists a buffer identifier $i$, a message $\amessage$ and
natural number $k$ such that (1) $a_j=i!\amessage$,
(2) $a_{j'}=i?\amessage$, (3) $a_j$ is the $k$-th send action
on $i$ in $e$, and (4) $a_{j'}$ is the $k$-th receive action
on $i$ in $e$.
A \emph{communication} of $e$ is either
a matching pair $\{j,j'\}$, or a singleton $\{j\}$ such that
$j$ does not belong to any matching pair (such a communication
is called unmatched). We write $\communicationsof{e}$ to denote the set of communications of $e$.

An execution imposes a total order on the actions. Sometimes,
however, it
is useful to visualise only the causal dependencies between 
actions. Message sequence
charts~\cite{msc-norm}
are usually used to this aim, as they only depict an order between matched pairs of actions and between actions of the same
process. However, message sequence charts do not represent graphically the causal dependencies due to shared buffers, like
the ones found in mailbox systems. 
Here we define \emph{action graphs} that depict all causal dependencies. When considering p2p communications, action graphs and message sequence charts coincide. 
We say that two actions $a_1,a_2$ commute if
$\processof{a_1}\neq\processof{a_2}$ and it is not the case
that $a_1$ and $a_2$ are two actions of the same type
on a same buffer:
there is no $\dag\in\{!,?\}$, $i\in I$ and 
$\amessage_1,\amessage_2\in\paylodSet$
such that $a_1=i\dag\amessage_1$ and $a_2=i\dag\amessage_2$.

\begin{definition}[Action graph]\label{def:agraph}
Given an execution
$e=a_1\cdots a_n$,  the 
\emph{action graph} $\actiongraphof{e}$ is the vertex-labeled
directed graph $(\{1,\ldots,n\},\prec_e,\lambda_e)$ where $\lambda_e(j)=a_j$ and $j\prec_{e} j'$ if (1) $j<j'$ and (2)
either $a_j,a_{j'}$ do not commute, or $\{j,j'\}$
is a matching pair. 
\end{definition}
When $j\prec_{e} j'$, we say that $a_j$ \emph{happens
before} $a_{j'}$.
Note that for an execution $e$ with $e\in\executionsof{\system}$,
$\prec_{e}^*$ is a partial order.
Two executions $e,e'$ are causally equivalent, denoted by $e\causalequiv{\system}e'$, if their action graphs
are isomorphic. Said differently, $e,e'$
are causally equivalent if they are two linearisations
of a same \emph{happens before} partial order.
Note that if $\Transition{\initconf}{e}{\aconf}{\system}$
and $e\causalequiv{\system} e'$, then
$\Transition{\initconf}{e'}{\aconf}{\system}$.
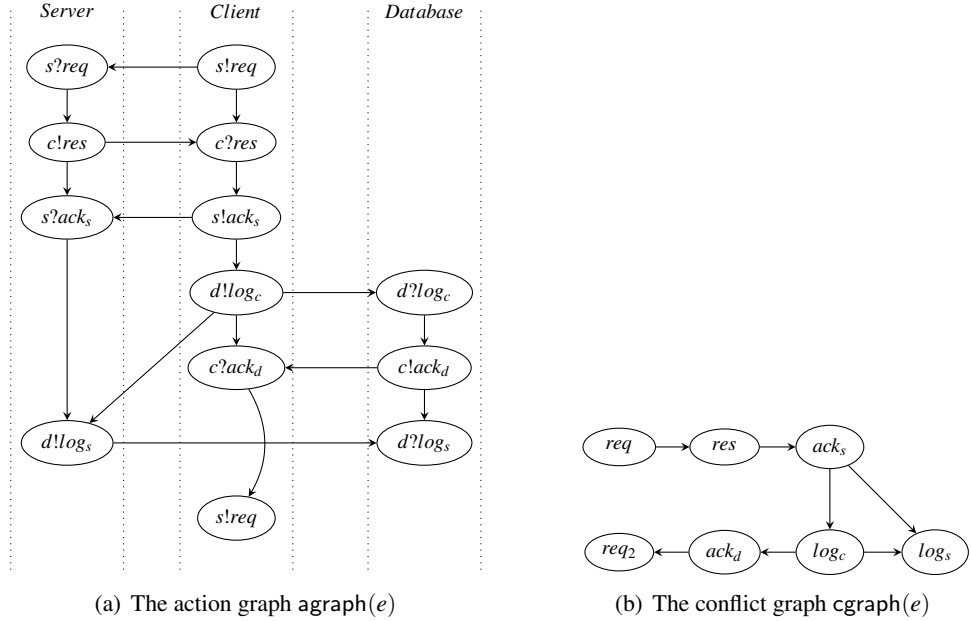
\begin{figure}
\centering
\subfigure[The action graph $\actiongraphof{e}$]{
\begin{tikzpicture}
	
	\def \xp {0.75}
	\def \xq {3}
	\def \xr {5.5}
	
	\def \ybot {-7.5}
	
	\node (c) at (\xp, 0) {$Server$};
	\node (s) at (\xq, 0) {$Client$};
	\node (d) at (\xr, 0) {$Database$};
	
	\node[elliptic state] (sreq1) at (\xq, -0.75) {$\csend{c}{s}{req}$};
	\node[elliptic state] (rreq1) at (\xp , -0.75) {$\crec{c}{s}{req}$};
	\node[elliptic state] (sres) at (\xp, -1.75) {$\csend{s}{c}{res}$};
	\node[elliptic state] (rres) at (\xq , -1.75) {$\crec{s}{c}{res}$};
	%\node[elliptic state] (rfull) at (\xq , -2.75) {$\crec{d}{s}{full}$};
%	\node[elliptic state] (sfull) at (\xr , -2.75) {$\csend{d}{s}{full}$};
	\node[elliptic state] (rack) at (\xp , -2.75) {$\crec{d}{s}{ack_s}$};
	\node[elliptic state] (sack) at (\xq , -2.75) {$\csend{c}{s}{ack_s}$};
 	\node[elliptic state] (slogc) at (\xq , -3.75) {$\csend{c}{d}{log_c}$};
	\node[elliptic state] (rlogc) at (\xr , -3.75) {$\crec{c}{d}{log_c}$};
	\node[elliptic state] (sackd) at (\xr , -4.75) {$\csend{d}{c}{ack_d}$};
	\node[elliptic state] (rackd) at (\xq , -4.75) {$\crec{d}{c}{ack_d}$};
	\node[elliptic state] (slog) at (\xp, -5.75) {$\csend{s}{d}{log_s}$};	
	\node[elliptic state] (rlog) at (\xr , -5.75) {$\crec{s}{d}{log_s}$};
	\node[elliptic state] (sreq2) at (\xq , -6.75) {$\csend{c}{s}{req}$};

	% Vertical lines, columns for each process
	\draw[dotted]
	(0,0) -- (0,\ybot)
	(\xp + \xp,0) -- (\xp + \xp,\ybot) % \xp vaut la moitié de la largeur d'une colone
	(\xq - \xp,0) -- (\xq - \xp,\ybot)
	(\xq + \xp,0) -- (\xq + \xp,\ybot)
	(\xr - \xp,0) -- (\xr - \xp,\ybot)
	(\xr + \xp,0) -- (\xr + \xp,\ybot);
	
	\tikzset{->,}
	\draw
	(sreq1) edge node{} (rres)
	(rres) edge node{} (sack)
%	(sack) edge node{} (sreq2)
	
	(rreq1) edge node{} (sres)
	(sres) edge node{} (rack)
	%(rfull) edge node{} (rack)
	(rack) edge node{} (slog)
	
	%(sfull) edge node{} (rlog)
	
	(sreq1) edge node{} (rreq1)
	(sres) edge node{} (rres)
%	(sfull) edge node{} (rfull)
	(sack) edge node{} (rack)
	
	(sack) edge node{} (slogc)
	(slogc) edge node{} (rackd)
	(slogc) edge node{} (rlogc)
	(rlogc) edge node{} (sackd)
	(sackd) edge node{} (rackd)
	
	(sackd) edge node{} (rlog)
	(slogc) edge node{} (slog)
	
	(rackd) edge[bend left] node{} (sreq2)
	
	(slog) edge node{} (rlog);
	
	%(sreq1) edge[bend left=12] node{} (sfull);
%	(sfull) edge node{} (sack);

\end{tikzpicture}
}\qquad
\subfigure[The conflict graph $\conflictgraph{e}$ ]{
\quad
\begin{tikzpicture}
		
	\node[elliptic state ] (req)  {${\ req_{\ }}$};
	\node[elliptic state, right of =req] (res)  {${\ res_{\ }}$};
%	\node[elliptic state, right of =res] (full)  {${full}$};

	\node[elliptic state, right of =res] (ack)  {${ack_s}$};
	\node[elliptic state, below of =ack] (log)  {${log_c}$};
	\node[elliptic state, right of =log] (logd)  {${log_s}$};
	\node[elliptic state, left of =log] (ackd)  {${ack_d}$};
	\node[elliptic state, left of =ackd] (req2)  {${req_2}$};

	\tikzset{->,}
	\draw
	(ack) edge node{} (log)
	(log) edge node{} (ackd)
	(log) edge node{} (logd)
	(req) edge node{} (res)
	(ack) edge node{} (logd)
	(ackd) edge node{} (req2)
	(res) edge node{} (ack);
\end{tikzpicture}
\quad
}
\caption{Causal dependencies of an execution of Client/Server/Database protocol}\label{fig:ex:agraphrunning}
\end{figure}

Another graphical tool that we will use to talk about equivalent executions is the \emph{conflict graph}, which is intuitively obtained from the action graph by merging
matching pairs of vertices. 
%A conflict graph highlights the precedences among communications.

\begin{definition}[Conflict graph]
Given an execution
$e=a_1\cdots a_n$, the conflict graph $\conflictgraph{e}$ of the execution $e$
is the directed graph $(\communicationsof{e},\conflictedge{e})$ where for all communications
$\acom_1,\acom_2\in\communicationsof{e}$,
$\acom_1\conflictedge{e} \acom_2$ if
there is $j_1\in\acom_1$ and $j_2\in\acom_2$
such that $j_1\prec_{e} j_2$.
\end{definition}

\begin{example}[Action and Conflict Graphs] \label{ex:actconfgraph}
We go back to the system depicted in Figure \ref{fig:ex:csd}. One of its executions is
$$e = \csend{c}{s}{req}\cdot\crec{c}{s}{req}\cdot\csend{s}{c}{res} \cdot \crec{s}{c}{res} \cdot
\csend{c}{s}{ack_s} \cdot \crec{c}{s}{ack_s} \cdot \csend{c}{d}{log_c} \cdot
\csend{s}{d}{log_s} \cdot \crec{c}{d}{log_c} \cdot \csend{d}{c}{ack_d} \cdot \crec{d}{c}{ack_d} \cdot \csend{c}{s}{req} \cdot \crec{s}{d}{log_s}$$
Figure \ref{fig:ex:agraphrunning}(a) shows $\actiongraphof{e}$.
Actions of the same process are represented vertically between the same dotted lines. As formally explained in Definition \ref{def:agraph}, an arc from an action $a$ and another $a'$ means that $a$ happens before $a'$.
To ease readability, the arcs that follow from transitivity 
are omitted. For example, in a given column, there should be an arc between every pair of actions.

Figure \ref{fig:ex:agraphrunning}(b) shows $\conflictgraph{e}$. To simplify the graph, instead of marking the matching pairs we simply identify them with the message  exchanged. Message $req_2$ represent the second send of $req$ in the execution above.

\qed
\end{example}

\section{Greedy systems}\label{sec:greedy-def}
% !TEX root = ice.tex
In this section we introduce greedy systems.
Those systems aim at mimicking rendez-vous or synchronous 
communications by checking whether each execution can be 
rescheduled to an equivalent one where all receptions 
immediately follow their corresponding send.

\begin{definition}[Greedy system]\label{def:onesynch}
    An execution $e$ is \emph{greedy} if all matching pairs
    are of the form $\{j,j+1\}$.
    A system $\system$ is \emph{greedy} if for all execution
    $e\in\executionsof{\system}$, there exists a greedy
    execution $e'$ such that $e\causalequiv{\system} e'$.
\end{definition}

\begin{example}
        The execution $s!req\cdot s?req\cdot c!res \cdot c?res \cdot s!ack_s \cdot d!log_c \cdot d?log_c$ is greedy,
        but the execution 
        $s!req\cdot s?req\cdot c!res \cdot c?res \cdot s!ack_s \cdot d!log_c \cdot s?ack_s$ is not greedy, although
        causally equivalent to a greedy execution.
        \qed
\end{example}

\begin{example}[Greedy system]\label{ex:greedysys}
    The system in Figure \ref{fig:ex:csd}  is greedy. Take for instance   execution $e$ of the example \ref{ex:actconfgraph}. This execution is not greedy as messages $log_c$, $log_s$, and $ack_d$ are not received right after their send. Still since those action can commute and by observing that the conflict graph in Figure~\ref{fig:ex:agraphrunning}(b) does not present any cycle (see Lemma~\ref{lem:greedy-graphical-caracterisation}) below), there exists an equivalent greedy execution $e'$:
    $$e' = \csend{c}{s}{req} \cdot \crec{c}{s}{req} \cdot \csend{s}{c}{res} \cdot \crec{s}{c}{res} \cdot
    \csend{c}{s}{ack_s} \cdot \crec{c}{s}{ack_s} \cdot \csend{c}{d}{log_c} \cdot
    \crec{c}{d}{log_c} \cdot \csend{d}{c}{ack_d} \cdot \crec{d}{c}{ack_d} \cdot \csend{c}{s}{req} \cdot \csend{s}{d}{log_s} \cdot \crec{s}{d}{log_s}.$$
    \qed
   \end{example}

\begin{example}\label{ex:counter-example-greedy}
Consider a system with two processes $p$ and $q$ each sending a message to the other, and whose corresponding receptions only happen  after the send
(see Figure \ref{fig:cgraphacyclic}). Then the execution
$e=p!\amessage_1\cdot q!\amessage_2 \cdot p?\amessage_1 \cdot q?\amessage_2$ is not causally equivalent to a greedy execution,
therefore the whole system is not greedy.
\qed
\end{example}

In the remainder of this section, 
we show that deciding whether a system
is greedy is feasible in polynomial time.
The proof is in three steps: first,
we give a graphical characterisation of the executions
that are causally equivalent to greedy executions;
second, we show that the non-greediness of a system
is revealed by the existence of \textquote{bad} executions
of a certain shape, called borderline executions.
Finally, we show that the graphical characterisation
can be exploited to show the regularity of the language
of borderline violations, from which we get the decidability of 
greediness.

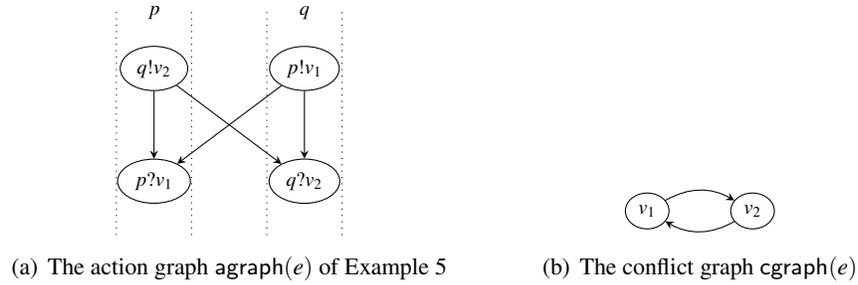
\begin{figure}
\centering
\subfigure[The action graph $\actiongraphof{e}$ of Example~\ref{ex:counter-example-greedy}]{
\qquad \qquad
\begin{tikzpicture}
	\node (p) at (0.5, 0) {$p$};
	\node (q) at (2.5, 0) {$q$};
	
	\node[elliptic state] (sv1) at (2.5,-0.75) {$\csend{q}{p}{\amessage_1}$};
	\node[elliptic state] (rv2) at (2.5,-2.25) {$\crec{p}{q}{\amessage_2}$};
	\node[elliptic state] (sv2) at (0.5,-0.75) {$\csend{p}{q}{\amessage_2}$};
	\node[elliptic state] (rv1) at (0.5,-2.25) {$\crec{q}{p}{\amessage_1}$};

	\draw[dotted]
	(2,0) -- (2,-3)
	(3,0) -- (3,-3)
	(0,0) -- (0,-3)
	(1,0) -- (1,-3);
	
	\tikzset{->,}
	\draw
	(sv1) edge node{} (rv2)
	(sv2) edge node{} (rv2)
	(sv2) edge node{} (rv1)
	(sv1) edge node{} (rv1);
\end{tikzpicture}
\qquad \qquad
} \qquad
\subfigure[The conflict graph $\conflictgraph{e}$]{
\qquad \quad
\begin{tikzpicture}
		
	\node[elliptic state] (sv1)  {${\amessage_1}$};
	\node[elliptic state, right of =sv1] (sv2)  {${\amessage_2}$};

	\tikzset{->,}
	\draw
	(sv1) edge[bend left] node{} (sv2)
	(sv2) edge[bend left] node{} (sv1);
\end{tikzpicture}
\qquad \quad
}
\caption{Visual representations of a non-greedy execution $e$}\label{fig:cgraphacyclic}
\end{figure}

\begin{lemma}\label{lem:greedy-graphical-caracterisation}
    % Let a system $\system=(\A_p)_{p\in\procSet}$ and an 
    % execution $e\in\executionsof{\system}$
    % be fixed. 
    %Then 
    $e$ is causally equivalent to a greedy 
    execution
    if and only if $\conflictgraph{e}$ is acyclic.
\end{lemma}

\begin{proof}
    The left to right implication follows from two observations:
    first, two causally equivalent executions have 
    isomorphic conflict graphs (because they have isomorphic
    action graphs), and second,
    the
    conflict graph of a greedy execution is acyclic, because 
    for a greedy execution
    $\acom_1\conflictedge{e} \acom_2$ induces
    $\min{\acom_1}<\min{\acom_2}$. Conversely,
    let $e=a_1\cdots a_n$, and
    assume that $\conflictgraph{e}$ is acyclic. 
    Let $\acom_1\ll\ldots\ll\acom_n$ be a topological order 
     on $\communicationsof{e}$,
    $e'=\acom_1\cdots\acom_n$ be the corresponding
    greedy execution, and  $\sigma$ be the permutation
    such that $e'=a_{\sigma(1)}\cdots a_{\sigma(n)}$. The claim is that $e\causalequiv{\system}e'$. Let $j,j'$ be two indices of $e$, and let us show that $j\prec_{e}j'$ iff $\sigma(j)\prec_{e'} \sigma(j')$. 
    First, $\{j,j'\}$ is a matching pair of $e$ if and only
    if $\{\sigma(j),\sigma(j')\}$ is a matching pair of $e'$, which shows the equivalence in that case.
    Assume now that $\{j,j'\}$ is not a matching pair of $e$, and let $\acom,\acom'$ be the communications of $e$
    containing $j$ and $j'$ respectively. Assume also that 
    $j\prec_e j'$,  we show that $\sigma(j)\prec_{e'}\sigma(j')$ (the other implication, similar, is omitted). First,
    $j\prec_e j'$ implies $\acom\conflictedge{e}\acom'$, which
    entails $\sigma(\acom)\conflictedge{e'}\sigma(\acom')$
    because $e$ and $e'$ have the same conflict graph.
    Moreover, $j\prec_e j'$ implies that 
    $a_{j}$ and $a_{j'}$ do not commute, therefore either 
    $\sigma(j)\prec_{e'}\sigma(j')$ or $\sigma(j')\prec_{e'}\sigma(j)$.
    By contradiction let 
     $\sigma(j')\prec_{e'}\sigma(j)$; then 
    $\sigma(\acom')\conflictedge{e'}\sigma(\acom)$, contradicting
    the acyclicity of the conflict graph of $e'$, which ends the proof.
\end{proof}

\begin{definition}[Borderline violations]
    An execution $e\in\executionsof{\system}$ is a 
    \emph{borderline violation} if (1) $e$ is not causally
    equivalent to a greedy execution, (2) $e=e_1\cdot i?\amessage$ for
    some \underline{greedy} execution $e_1$ and receive action $i?\amessage$.
    \end{definition}
    
    \begin{example}[Borderline violation]
    An example of a borderline violation for the system whose
    unique maximal execution is the one of
    Figure~\ref{fig:cgraphacyclic}  is the execution  
    $$e = \send{p}{q}{\amessage_2} \cdot \send{q}{p}{\amessage_1} \cdot \rec{q}{p}{\amessage_1} \cdot \rec{p}{q}{\amessage_2}.$$
    Figure \ref{fig:cgraphacyclic} shows its action and conflict  graph.
    The action graph makes it easy to see that any execution $e'$ equivalent to $e$ will require both the send actions to be done before the first reception, therefore at least one reception will not follow its matching send action.
    %The cycle in the conflict graph shows precisely this impossibility.
    \qed
    \end{example}

\begin{lemma}\label{lem:greedy-equals-no-borderline-violation}
    $\system$ is greedy if and only if $\executionsof{\system}$
    contains no borderline violation.
\end{lemma}

\begin{proof}
Obviously, if $\executionsof{\system}$ 
contains a borderline violation, $\system$ is not greedy.
Conversely, assume that $\system$ is not greedy, and let us
show that $\executionsof{\system}$ contains a borderline violation. Let $e\in\executionsof{\system}$ be an execution
that is not causally equivalent to a greedy execution and
of minimal length among all such executions. Then $e=e_1\cdot a$
with $e_1$ causally equivalent to a greedy execution.
Let $e_1'$ be a greedy execution causally equivalent to $e_1$.
Then $e'=e_1'\cdot a \in\executionsof{\system}$. Moreover,
if $a$ is a send action, then $e'$ is greedy, contradicting the fact that $e$ is not causally equivalent to a greedy execution.
Therefore, $e'$ is a borderline violation.
\end{proof}

Let $\Sigma=I\times \{!,!?\}\times \paylodSet$ denote the
set of communications, and let $\Sigma_{?}=I\times\{?\}\times\paylodSet$
be the set of receive actions. 
Then a greedy execution
can be represented by a word in $\Sigma^*$
and a borderline 
violation is represented by a word in $\Sigma^*\cdot\Sigma_{?}$.
So now, we define two non-deterministic finite state automata
over $\Sigma\cup\Sigma_{?}$: the first one accepts all greedy executions of a system, and the second one all borderline violations. We later explain
how these automata are used to decide greediness.
%Clearly, if, by concatenating the greedy executions with a reception we obtained an execution that is present among the borderline ones  then for previous lemma the system is not greedy. 

\begin{lemma}\label{lem:greedy-execution-is-regular}
    Let $\system=(\A_p)_{p\in\procSet}$ of size $n$ and
    $\paylodSet=\bigcup_{p\in\procSet} \paylodSet_p$, 
    $I=\bigcup_{p\in\procSet}I_p$ be fixed.
    There is a non-deterministic finite state automaton $\A_{gr}$
    computable in time $\bigO{|\paylodSet||I|^22^{|I|}n^{|\procSet|+2}}$
    such that $\languageof{\A_{gr}}=\{e\cdot i ?\amessage\in\Sigma^*\cdot\Sigma_{?}\mid e\cdot i?\amessage\in\executionsof{\system}\mbox{ and $e$ is greedy}\}$.
\end{lemma}

\begin{proof}
    Let $\productof{\system}=(L_{\system},\paylodSet,I,\actionSet,\delta_{\system},\vect{l}_0)$
    and let $\A_{gr}=(L_{gr},\delta_{gr},l_{gr,0},F_{gr})$ be 
    the non-deterministic finite state automaton over $\Sigma$
    with
    $L_{gr}=L_{\system}\times (\{\epsilon\}\cup I \times\amessage)\times 2^I\cup\{l_{F}\}$, 
    $l_{gr,0}=(\vect{l}_0,\epsilon,\emptyset)$,
    $F_{gr}=\{l_F\}$, and the transitions defined as follows.
    First, while reading a letter $c\in\Sigma$, 
    $\big((\vect{l},x,S),c,(\vect{l}',x',S')\big)\in\delta_{gr}$ 
    if
    \begin{itemize}
        \item 
        $\Transition{(\vect{l},\vect{b})}{c}{(\vect{l}',\vect{b'})}{\system}$ 
        for some $\vect{b},\vect{b'}$ 
        such that for all $i\in I$, $b_i\neq\emptyset$ iff $i\in S$, and $b_i'\neq\emptyset$ 
        iff $i\in S'$, and
        \item one of the two following holds
        \begin{itemize}
            \item 
            either $x=x'$,  
            \item or $x=\epsilon$ and $x'=(i,\amessage)$ and $c=i!\amessage$ and $i\not\in S$.
        \end{itemize}
    \end{itemize} 
    Second, while reading the letter $i?\amessage\in\Sigma_{?}$,
    $\big((\vect{l},x,S),i?\amessage,l_F\big)\in\delta_{gr}$ if 
    $x=(i,\amessage)$ and $(\vect{l},i?\amessage,\vect{l}')\in\delta_{\system}$ for some
    $\vect{l}'\in L_{\system}$.
    Then 
    $\languageof{\A_{gr}}$ is as announced.
    Moreover, each transition of 
    $\A_{gr}$ can be constructed in constant time,
    so $\A_{gr}$ can be constructed in time as announced.
\end{proof}

\begin{lemma}\label{lem:borderline-violation-is-regular}
    There is a non-deterministic finite state automaton $\A_{bv}$
    computable in time $\bigO{|I|^3|\paylodSet|^3}$    
    %$\bigO{|I|^4||\paylodSet|^4n^{|\procSet|+2}}$ 
    such that 
    $\languageof{\A_{bv}}=\{e\in\Sigma^*\cdot\Sigma_?\mid \conflictgraph{e}
    \mbox{ contains a cycle}\}$.
\end{lemma}

\begin{proof}
    Let $\A_{bv}=(L_{bv},\delta_{bv},l_{bv,0},\{l_{bv,1}\})$ be the 
    non-deterministic finite state automaton over 
    $\Sigma\cup \Sigma_{?}$ such that $L_{bv}=\{l_{bv,0},l_{bv,1}\}\cup
    \Sigma_{?}\times\Sigma$, and for all $c,c'\in\Sigma$,
    for all $a\in\Sigma_{?}$, for all $i\in I$, $\amessage\in\paylodSet$,
    (1) $(l_{bv,0},c,l_{bv,0})\in\delta_{bv}$
    (2) $(l_{bv,0},i!\amessage,(i?\amessage,i!\amessage))\in\delta_{bv}$,
    (3) $((a,c),c',(a,c))\in\delta_{bv}$ 
    (4) $((a,c),c',(a,c'))\in\delta_{bv}$ if 
    $\processof{c}\cap\processof{c'}\neq\emptyset$,
    and (5) $((i?\amessage,c),i?\amessage,l_{bv,1})\in\delta_{bv}$ 
    if $\processof{c}\cap\processof{i?\amessage}\neq\emptyset$.
    Then $\languageof{\A_{bv})}=\{e\in \Sigma^*\Sigma_{?}\mid \conflictgraph{e}\mbox{ contains a cycle}\}$. Moreover, each
    transition of $\A_{bv}$ can be constructed in constant time,
    so $\A_{bv}$ can be constructed in time as announced.
\end{proof}

From the computability of the
two previous automata, we deduce the decidability of system greediness.

\begin{theorem}\label{thm:greedy-is-decidable-in-ptime}
    Whether a system $\system=(\A_p)_{p\in\procSet}$ of size $n$ is greedy is decidable
    in time $\bigO{|I|^5|\paylodSet|^42^{|I|}n^{|\procSet|+2}}$.
\end{theorem}

\begin{proof} 
    Let $\A_{gr}$ and $\A_{bv}$ be the two automata defined
    in Lemmas~\ref{lem:greedy-execution-is-regular} and 
    \ref{lem:borderline-violation-is-regular}.
    By Lemma~\ref{lem:greedy-graphical-caracterisation} and
    by definition of a borderline violation, the set of
    borderline violations of $\system$
    is
    $\languageof{\A_{gr}}\cdot\Sigma_{?}\cap \languageof{\A_{bv}}$.
    So, by Lemma~\ref{lem:greedy-equals-no-borderline-violation},
    $\system$ is greedy if and only if 
    $\languageof{\A_{gr}}\cdot\Sigma_{?}\cap \languageof{\A_{bv}}=\emptyset$.
    The claim then directly follows from the fact that
    emptiness testing for a non-deterministic finite state 
    automaton of size $n$ is in time $\bigO{n}$.
\end{proof}

\section{Model-Checking Greedy Systems}\label{sec:model-checking}
% !TEX root = ice.tex

In this section we explore how to verify
various safety properties of greedy systems in polynomial
time. Since the reachability set of a greedy
system is not regular, it is not obvious that regular
safety properties are always decidable. We show
that this problem actually is decidable, with
a polynomial time complexity under mild assumptions.
Then, we list a few regular safety
properties that were also considered in other works,
in particular for approaches based on session 
types~\cite{DBLP:journals/jacm/HondaYC16,DBLP:conf/esop/DenielouY12,DBLP:journals/pacmpl/ScalasY19}.

\subsection{Checking Regular Safety Properties}
Let 
$\system=(\A_p)_{p\in\procSet}$ with 
$I=\bigcup_{p\in\procSet}I_p=\{1,\ldots,|I|\}$ 
be a system. 
We identify a word $w=\vect{l} \cdot \sharp \cdot b_1\cdot \sharp\cdots\sharp \cdot b_{|I|}\in L_{\system}\cdot (\sharp\cdot\paylodSet^*)^{|I|}$ with the
configuration $(\vect{l},b_1,\ldots,b_{|I|})$. We say that a set of 
configurations $P(\system)$ is \emph{regular}\footnote{also
called \emph{channel recognizable} by Cece and Finkel~\cite[Definition~10]{DBLP:journals/iandc/CeceF05}} if the 
corresponding set of words coding these configurations is regular.
A \emph{property} is a function $P$ that associates to every
system $\system$ a set of configurations $P(\system)$.
We say that $P$ is regular if $P(\system)$ is regular for 
all $\system$, and \emph{computable} in time 
$\bigO{f(n)}$ if a non-deterministic 
finite state automaton $\A$ accepting 
$P(\system)$ can be computed in time $\bigO{f(|\system|)}$.
The $P$ \emph{safety problem} is whether a system
$\system$ is such that 
$\reachablesetof{\system}\cap P(\system)= \emptyset$.
Examples of safety problems are discussed below in 
Section~\ref{sec:examples-of-regular-safety-properties}.

Cece and Finkel showed that, 
for a binary half-duplex system $\system$, 
$\reachablesetof{\system}$ is regular and computable in 
polynomial time~\cite[Theorem~26]{DBLP:journals/iandc/CeceF05}. 
Since the emptiness of the intersection
of two polynomial time computable regular languages is 
decidable in polynomial time, for any regular polynomial
time property $P$, the $P$ safety problem is
decidable in polynomial time for binary
half-duplex systems. 

For greedy systems, however,
the situation is a bit different. Indeed, observe that 
$\reachablesetof{\system}$ may be non-regular and even 
context sensitive (it is easy to define a system
with one machine and 3 buffers, such that for some control state 
$l$, $\reachablesetof{\system}\cap l(\sharp\paylodSet^*)^3=
\{l\sharp a^n\sharp b^n\sharp c^n\mid n\geq 0\}$).
So it is not obvious how to decide the emptiness
of $\reachablesetof{\system}\cap P(\system)$.

Still, the $P$ safety problem remains decidable
for a computable regular property $P$. 

\begin{theorem}\label{thm:regular-model-checking-is-in-PTIME}
    Let $P$ be a computable regular 
    property. Then, for a given
    greedy system $\system=(\A_p)_{p\in\procSet}$, 
    the $P$ safety problem is
    decidable. Moreover,
    if $P$ is
    computable in time 
    $\bigO{|\system|^k}$ for some $k\geq 0$, then
    the problem is decidable in time 
    $\bigO{|\system|^{k+|\procSet|+2}}$.
\end{theorem}

\begin{proof}
    Let $\system$ be fixed, with 
    $\productof{\system}=(L_{\system},\paylodSet,I,\actionSet,\delta_{\system},\vect{l}_{\system,0})$.
    Let $\A=(L_{\A},\delta_{\A},l_{\A,0},F_{\A})$ be the
    polynomial time computable non-deterministic
    finite state automaton over alphabet 
    $L_{\system}\cup\{\sharp\}\cup\paylodSet$ such that 
    $\languageof{\A}=P(\system)$.
    We define an automaton $\A_P=(L_P,\delta_P,L_{P,0},F_P)$ 
    over the alphabet
    $\Sigma$ of communications such that for
    all greedy execution $e\in\executionsof{\system}$,
    $e\in \languageof{\A_P}$ iff there is 
    $\aconf\in P(\system)$ such that 
    $\Transition{\initconf}{e}{\aconf}{\system}$.

    Before defining $\A_p$ formally, let us give
    some intuitions about how it works on an example.
    Assume that $\A_p$ reads $e=1!a\cdot 2!?b\cdot 2!c\cdot 1!d$,
    and that the final configuration $\aconf$ is 
    $\vect{l}\sharp ad\sharp c$. While reading $e$,
    $\A_p$ should check the existence of an accepting run
    of $\A$ on $\aconf$.
    When $\A_P$ reads a 
    communication, there are two cases. Either the communication
    is a matched send (like $2!?b$) 
    and therefore it does not contribute
    to the final configuration, so $\A_p$ merely ignores it.
    Or the communication is an unmatched send, and it contributes
    to a piece of the accepting run of $\aconf$ on $\A$.
    However, these pieces of the accepting run of $\A$ are not 
    necessarily consecutive. For instance, in the above 
    execution, $a$ and $d$ are consecutive in the run of $\A$,
    but $\A_p$ reads $c$ in between, although $c$ contributes 
    only later to the run of $\A$.
    To correctly check the existence of a run of $\A$ on 
    $\aconf$, $\A_p$ uses for each buffer a distinguished "pebble" placed
    on a state of $\A$. Every time $\A_p$ reads 
    an unmatched send $i!\amessage$, it moves the $i$-th pebble 
    along an $\amessage$ transition of $\A$. So each pebble checks for a
    piece of the accepting run of the whole word coding
    the final configuration. $\A_p$ therefore also needs to make
    sure that all of these pieces of runs can be concatenated to
    form a run of $\A$. Therefore, at the beginning, $\A_p$
    guesses an initial control state $l_i$ and
    a final control state $l_i'$ for each pebble, and ensures
    that $(l_i',\sharp,l_{i+1})\in\delta_{\A}$.
    While doing so, going back to our example,
    $\A_p$ ensures that $ad\#c$ can be read by $\A$.
    It remains also to deal with the control state: indeed,
    $\A$ should accept $\vect{l}\sharp ad\sharp c$.
    So $\A_P$ guesses before reading $e$
    that the
    control state will be $\vect{l}$ after executing $e$, 
    and while reading
    $e$, it computes the current control state of $\system$.
    In the end, it checks that this control state actually is
    $\vect{l}$.
    
    Now that we presented the intuitions about $\A_p$, let us
    define it formally.
    Let us start with the set of control states. Let
    $L_P=L_{\system}\times L_{\system}\times L_{\A}^{|I|}\times L_{\A}^{|I|}$.
    Intuitively, the control state $(\vect{l}_{\system},\vect{l}_{F},\vect{l}_{\A},\vect{l}_{I})$ of $\A_p$ corresponds to
    a situation where: 
    (1)~the current control state of $\system$ is 
    $\vect{l}_{\system}$, 
    (2)~the guessed final control state of
    $\system$ is $\vect{l}_{F}$, 
    (3)~the $i$-th pebble currently
    is on state $l_{\A,i}$ of $\A$, and 
    (4)~$\vect{l}_{I}$ is a copy of the initial positions
    of the pebbles and will be checked against their final
    positions in the end to ensure that all pieces of runs
    can be concatenated.

    Let us now define the set $L_{P,0}$ of initial control 
    states of $\A_P$. 
    Let us set that
    $(\vect{l}_{\system},\vect{l}_{F},\vect{l}_{\A},\vect{l}_{I}\in L_{P,0}$ if 
    (1)~$\vect{l}_{\A}=\vect{l}_{I}$,
    (2)~$l_{\A,1}\in\delta_{\A}^*(l_{\A,0},\vect{l}_{F}\cdot\sharp)$, and 
    (3)~$\vect{l}_{\system}=\vect{l}_{\system,0}$. 
    Intuitively, a control state is initial if 
    (1)~$\vect{l}_{I}$ is a copy of $\vect{l}_{\A}$, 
    (2)~the position of the pebble of buffer $1$ is on a state
    that is reachable in $\A$ after reading $\vect{l}_{F}\sharp$ and 
    (3)~$\vect{l}_{\system}$ is the initial control state of $\system$.

    Similarly, let us now define the set $F_{P}$ of final 
    control states of $\A_P$. 
    Let us set $(\vect{l}_{\system},\vect{l}_{F},\vect{l}_{\A},\vect{l}_{I})\in F_{P}$ if 
    (1) $\vect{l}_{\system}=\vect{l}_{F}$, 
    (2) for all $i=1,\ldots,|I|-1$, $(l_{\A,i},\sharp,l_{\A,i+1})\in\delta_{\A}$, and 
    (3) $l_{\A,|I|}\in F_{\A}$.
    Intuitively, a control state is final if 
    (1)~the current control state of $\system$ corresponds to 
    the guessed final one,
    (2)~the $i$-th pebble can be moved along a $\sharp$ 
    transition so as to reach the initial position of pebble $i+1$, and 
    (3) the pebble of the last buffer reached an accepting state
    of $\A$. 

    Let us now define the set $\delta_P$ 
    of transitions of $\A_p$. Let us set that
    $$
        \Big((\vect{l}_{\system},\vect{l}_{F},\vect{l}_{\A},\vect{l}_{I}),~\acom,~(\vect{l}_{\system}',\vect{l}_{F}',\vect{l}_{\A}',\vect{l}_{I}')\Big)~\in~\delta_P
    $$
    if
    (1)~$\vect{l}_{F}=\vect{l}_{F}'$,
    (2)~$\vect{l}_I=\vect{l}_{I}'$,
    (3)~$\vect{l}_{\system}'\in\delta_{\system}^*(\vect{l}_{\system},\acom)$,
    (4.1)~if $\acom=i!?\amessage$, then $\vect{l}_{\A}=\vect{l}_{\A}'$, and
    (4.2)~if $\acom=i!\amessage$, then $(l_{\A,i},\amessage,l_{\A,i}')\in\delta_{\A}$ and $l_{\A,j}=l_{\A,j}'$
    for all $j\neq i$.
    Intuitively, when it reads a matched send $i!?\amessage$,
    $\A_p$ only updates $\vect{l}_{\system}$ according to
    the sequence of actions $i!?\amessage$, while when
    it reads an unmatched send $i!\amessage$ it also updates the
    position of the $i$-th token.
    
    Now that $\A_P$ is defined, observe that
    $\reachablesetof{\system}\cap P(\system)=\emptyset$
    iff $\languageof{\A_p}\cap \languageof{\A_{gr}}=\emptyset$, where $\A_{gr}$ is the automaton defined
    in Lemma~\ref{lem:greedy-execution-is-regular}. The
    emptiness of this intersection is decidable in time 
    $\bigO{|\A_P|\cdot|\A_{gr}|}$, which shows the claim.
\end{proof}

\subsection{Examples of Regular Safety Problems}\label{sec:examples-of-regular-safety-properties}

In this section we review a few properties of systems
that are polynomial-time computable regular properties
and showcase some applications of Theorem~\ref{thm:regular-model-checking-is-in-PTIME}.

\paragraph{Reachability.} The control state
reachability problem is to decide, given a system $\system$
and a control state $\vect{l}\in L_{\system}$,
whether there exists $\vect{b}\in (\paylodSet^*)^I$ and $e\in\actionSet^*$ such that
$\Transition{\initconf}{e}{(\vect{l},\vect{b})}{\system}$.
The configuration reachability problem, on the other hand,
is to decide, given a system $\system$ and a configuration
$\aconf$, whether $\aconf\in\reachablesetof{\system}$.
Both problems are safety problems for a regular property $P$
computable in polynomial time (and even constant time):
$P(\system)=\vect{l}\cdot(\sharp\cdot\paylodSet^*)^{|I|}$
for the control state reachability problem, and $P(\system)=\{\aconf\}$ for the configuration reachability problem.

\paragraph{Unspecified reception.}
Unspecified receptions is one of the errors
that session types usually forbid. This
error makes more sense for mailbox systems, 
so let us assume for now that $\system$ is a mailbox system.
A configuration is an \emph{unspecified reception} if 
one of the participants is in a receiving state, and none of its outgoing transitions can receive the first message in its buffer. 

Let us define these notions more formally. A control state
$l_p$ of process $p$ is a receiving state if for
all $a,l'$ such that $(l_p,a,l')\in\delta_p$, $a$ is
a receive action.
The set $\{\amessage\mid (l_p,p?\amessage,l')\in\delta_p \mbox{ for some }l'\}$ is called the ready set of $l_p$.

%\begin{definition}[Unspecified reception configuration]\label{def:unspec}
A configuration $(\globalstate{l}, \B)$ is said an \emph{ unspecified reception configuration} if there is $p \in\procSet$ such that (1) $l_p$ is a receiving state, (2) 
$b_p=\amessage b'$ for some $\amessage\in\paylodSet$ and $b'\in\paylodSet^*$, and (3)
$\amessage$ is not in the ready set of $l_p$.
%\end{definition}
It can be observed that the set $UR(\system)$ of 
unspecified receptions of $\system$ defines a regular
property that is computable in polynomial time.

\paragraph{Progress.}
Another property that is central in session
types is progress. A global control state $\vect{l}$
of $\system$
is \emph{final} if there is no action $a$ and global control state $\vect{l}'$ such that 
$(\vect{l},a,\vect{l}')\in\delta_{\system}$.
A configuration $\aconf=(\vect{l},\vect{b})$ 
of $\system$ \emph{satisfies progress} if either $\vect{l}$
is final or
there is a configuration
$\aconf'$ and an action $a$
such that $\transition{\aconf}{a}{\aconf'}{\system}$.
A system satisfies progress if all reachable configurations
satisfy progress. It can be observed
that the set $NP(\system)$ of configurations that do not
satisfy progress is regular and polynomial time computable.

\subsection{Boundedness}\label{sec:boundedness}
There are other examples of properties
that are regular and polynomial time, but some interesting
ones are not safety properties. To conclude
this section, we consider one of these properties.

\begin{definition}[Boundedness]\label{def:boundedness}
    Let $\system$ be a system and $k\geq 0$. 
    A channel $i\in I$ is $k$-\emph{bounded} if for all $\decconf \in \erress{\system}$ $|b_i| \leq k$.
    $\system$ is $k$-bounded if for all $i \in I$, $i$ is $k$-bounded.
    \end{definition}
    
    \begin{theorem}
        Whether there exists $k\geq 0$ such that
        a greedy system $\system$ is $k$-bounded
        is decidable in polynomial time. Moreover, $k$
        is computable in polynomial time.
    \end{theorem}
    
    \begin{proof}
        Let $\Sigma_{!}\subseteq \Sigma$ be the set of unmatched
        communications, and $\sigma:\Sigma^*\to\Sigma^*_{!}$ the
        morphism that erase all matched communications.
        By Lemma~\ref{lem:greedy-execution-is-regular}, and by
        the closure of regular languages under morphisms,
        there is a polynomial time computable non-deterministic 
        finite state 
        automaton $\A$ such that 
        $\languageof{\A}=\{\sigma(e)\mid e \in\executionsof{\system}
        \mbox{ is greedy}\}$. 
        Then $\system$ is $k$-bounded for some $k$
        if and only if $\languageof{\A}$ is finite, or 
        equivalently if and only if $\A$, once pruned
        (removing states that are not reachable from the initial 
        state and co-reachable from a final state), is acyclic.This is decidable in time $\bigO{|\A|}$,
        and the maximal length $k$ of a word of 
    $\languageof{\A}$ also is computable in time $\bigO{|\A|}$.
\end{proof}
%\hrule
%\ELNOTE{Ci-dessous: incorporer dans la partie au-dessus.}
%\input{model-checking.tex}

\section{Mailbox Half-Duplex Systems}\label{sec:mailbox-hd}
% !TEX root = ice.tex

Binary half-duplex systems, (called
simply half-duplex by Cece and 
Finkel~\cite{DBLP:journals/iandc/CeceF05}) are binary
systems such that all reachable configurations
$(l_1,l_2,b_1,b_2)$ are such that either $b_1=\epsilon$
or $b_2=\epsilon$. In the previous
sections, we established that greedy systems
enjoy the same decidability and complexity results
as binary half-duplex systems.
In this section, we defend the claim that greedy systems 
could also merit the name of multiparty \halfduplex systems.

First, observe that binary half-duplex systems are
greedy (see~\cite[Lemma~20]{DBLP:journals/iandc/CeceF05}).
The converse does not hold in general: some binary greedy
systems are not half-duplex. However, under an extra
hypothesis, both are equivalent. A system is called
\emph{without orphan messages} if for all reachable
configuration containing a message in a buffer,
it is possible to reach a configuration where this message
has been received. This property is also enforced
by session types and is very natural in communicating systems.
Then observe that, for a given a binary system $\system$
without orphan messages, the following two are equivalent:
(1) $\system$ is binary half-duplex, and (2) $\system$ is 
greedy.

Let us now consider multiparty systems. 
Cece and Finkel proposed two notions of multiparty
half-duplex systems, but conclude that they were not
well behaved (one being too restrictive, and the other
Turing powerful).
Both of these generalisations relied on peer-to-peer communication.  
We propose to consider mailbox communication instead. 

\begin{definition}[Half-duplex execution]%\label{def2}
	Fix a mailbox system $\system$.
	An execution 
	$e = \transition{(\globalstate{l_0}, \B_0)}{a_1}{(\globalstate{l_1}, \B_1)}{}\xrightarrow{}\cdots \xrightarrow{a_n}(\globalstate{l_n}, \B_n)$
	is \emph{\halfduplex} if for all 
	$j=1,\ldots,n$, if $a_j$ is a send action, then $b_p^{i-1} = \varepsilon$, where $p=\processof{a_j}$.
\end{definition}
Intuitively, an execution is half-duplex if every process
empties its queue of messages before sending.

\begin{definition}[Mailbox half-duplex system]
	A mailbox system $\system$ is mailbox \halfduplex if for all 
	execution $e\in\executionsof{\system}$, there is
	a half-duplex execution $e'$ such that $e\causalequiv{\system} e'$.
\end{definition}

\begin{example}
The system in Figure \ref{fig:ex:csd} is \halfduplex.
Indeed even if execution $e$ of Example \ref{ex:actconfgraph} is not \halfduplex, by considering one of its greedy equivalents $e''$:
$$e'' = \csend{c}{s}{req} \cdot \crec{c}{s}{req} \cdot \csend{s}{c}{res} \cdot \crec{s}{c}{res} \cdot
\csend{c}{s}{ack_s} \cdot \crec{c}{s}{ack_s} \cdot \csend{c}{d}{log_c} \cdot
\crec{c}{d}{log_c} \cdot \csend{d}{c}{ack_d} \cdot \crec{d}{c}{ack_d}  \cdot \csend{s}{d}{log_s} \cdot \crec{s}{d}{log_s}\cdot \csend{c}{s}{req}$$
 we obtain a \halfduplex execution. 
 Notice that execution $e'$ of Example~\ref{ex:greedysys} is greedy but not \halfduplex   as the send of message $log_s$ is done when the buffer of the Server is filled with message $req$. \qed
\end{example}

Note that 
a binary system is mailbox half-duplex if and only
if it is binary half-duplex. In the remainder, we therefore
sometimes say simply half-duplex instead of binary half-duplex
or mailbox half-duplex.

\begin{theorem}\label{th:1synch}
Mailbox half-duplex systems are greedy.
\end{theorem}

\begin{proof}
We reason by contradiction. 
Assume that $\system$ is not greedy, we show that $\system$ is not half-duplex.
Let $e\in\executionsof{\system}$ be any execution that
is not causally equivalent to a greedy execution
(for instance, take for $e$ a borderline violation).
We claim that for all $e'$ such that $e\causalequiv{\system} e'$,
$e'$ is not half-duplex.  
Since $e'$ is not causally equivalent to a greedy
execution, we get, by Lemma~\ref{lem:greedy-graphical-caracterisation}, that
$\conflictgraph{e'}$ contains a cycle of communications
$c_1\conflictedge{e'} c_2\conflictedge{e'}\ldots\conflictedge{e'} c_n\conflictedge{e'} c_1$ where for all $i=1,\ldots n$,
either $c_i=\{j_i,k_i\}$ is a matching pair,
or $c_i=\{j_i\}$ is an unmatched send. We assume that
$j_i<k_i$, i.e. $j_i$ is the index of the send action
and $k_i$ the index of the receive action.
Up to a circular permutation, we can also assume, 
without loss of generality, 
that $j_1$ is the first send among them
in $e'$, i.e. $j_1<j_{\ell}$ for all $\ell=2,\ldots,n$.
Now, let us reason by case analysis on the nature of the
conflict edge $c_n\conflictedge{e'}c_1$.
\begin{itemize}
	\item case ``$c_n\lconflictedge{SS} c_1$": $j_n \prec_{e'}
	j_1$. Then $j_n<j_1$, contradicts the minimality of $j_1$.
	Impossible.
	\item case ``$c_n\lconflictedge{RS} c_1$": $k_n \prec_{e'}
	k_1$. Then $j_n<k_n<j_1$, impossible.
	\item case ``$c_n\lconflictedge{RR} c_1$": $k_n \prec_{e'}
	k_1$. Then $k_n<k_1$ and either (1)
	$\processof{a_{k_n}}=\processof{a_{k_1}}$ 
	or (2) there is $i\in I$, $\amessage,\amessage'\in\paylodSet$
	such that $a_{k_n}=i?\amessage$ and $a_{k_1}=i?\amessage'$.
	Because of the mailbox semantics, (1) and (2) are equivalent,
	so (2) is granted. But then $a_{j_n}=i!\amessage$
	and $a_{j_1}=i!\amessage'$. Since $e'$ is a FIFO
	execution, and $k_n<k_1$, we get that
	$j_n<j_1$, and again the contradiction.
	\item case ``$c_n\lconflictedge{SR} c_1$": $j_n \prec_{e'}
	k_1$. Then $j_n<k_1$, and $\processof{a_{j_n}}=\processof{a_{k_1}}$. Moreover, $j_1<j_n$ by the minimality of $j_1$.
	To sum up, let $p,q,r,\amessage_1,\amessage_2$ be such
	that $a_{j_1}=\send{p}{q}{\amessage_1}$, 
	$a_{k_1}=\rec{p}{q}{\amessage_1}$, and $a_{j_n}=\send{q}{r}{\amessage_2}$. Then we just showed that 
	$e'=\ldots \send{p}{q}{\amessage_1}\ldots\send{q}{r}{\amessage_2}\ldots\rec{p}{q}{\amessage_1}\ldots$,
	so $e'$ is not a half-duplex execution.
\end{itemize}
\end{proof}

\noindent
\begin{minipage}{.75\columnwidth}
\qquad Notice that the converse of Theorem~\ref{th:1synch} does
not hold:
being greedy is not a sufficient condition to be \halfduplex.
Indeed, an unmatched send can fill the buffer of a process willing to send. More precisely, consider
the \MSC on the right. It depicts a greedy system that is not \halfduplex: in fact the buffer for $q$ is not empty when $\amessage_2$ is sent.
We conjecture that this is the only
pathological situation, and that, like in the binary setting, if
a system is greedy and has no orphan messages, then it is
half-duplex.
\end{minipage}
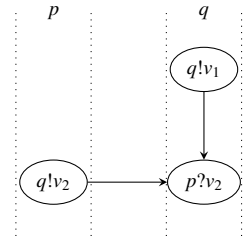
\begin{minipage}{.25\columnwidth}
\centering
\begin{tikzpicture}
	\node (p) at (0.5, 0) {$p$};
	\node (q) at (2.5, 0) {$q$};
	
	\node[elliptic state] (sv1) at (2.5,-0.75) {$\csend{p}{q}{\amessage_1}$};
	\node[elliptic state] (rv2) at (2.5,-2.25) {$\crec{q}{p}{\amessage_2}$};
	\node[elliptic state] (sv2) at (0.5,-2.25) {$\csend{p}{q}{\amessage_2}$};
		
	\draw[dotted]
	(2,0) -- (2,-3)
	(3,0) -- (3,-3)
	(0,0) -- (0,-3)
	(1,0) -- (1,-3);
	
	\tikzset{->,}
	\draw
	(sv1) edge node{} (rv2)
	(sv2) edge node{} (rv2);

\end{tikzpicture}

\end{minipage}
\vspace{0,2cm}

\begin{figure}[t]
	\centering
	\begin{tikzpicture}
	\tikzset{->,}
	% Client
	\node[state, initial] (c0) {$0$};
	\node[state, right=1.8cm of c0] (c2) {$2$};
	\node[state, above of=c0] (c1) {$1$};
	\node[left=1cm of c1] (c) {\Large \underline{$Client$}};
	\node[state, below=1.8cm of c2] (c3) {$3$};
	\node[state, below=1.8cm of c0] (c4) {$4$};
	
	\node[below=0.9cm of c2] (spacer1) {};
	\node[right=2.7cm of spacer1] (spacer) {};
	% Seller
	\node[state, initial, right of=spacer] (s0) {$0$};
	\node[state, above=2cm of spacer] (s1) {$1$};
	\node[state, right=1.2cm of s0] (s2) {$2$};
	\node[state, right=2.1cm of s2] (s3) {$3$};
	\node[state, below of=s2] (s4) {$4$};
	\node[state, above of=s2] (s5) {$5$};
	\node[above=0.8cm of s5] (s) {\Large \underline{$Seller$}};
	
	% Bank
	\node[below=2cm of c4] (b) {\Large \underline{$Bank$}};
	\node[state, initial, right of=b] (b0) {$0$};
	\node[state, right=2.1cm of b0] (b1) {$1$};
	\node[state, right=2cm of b1] (b2) {$2$};
	\node[state, right=2cm of b2] (b3) {$3$};
	
	\draw
	(c0) edge[left, bend left] node{$\csend{c}{s}{ask\_price}$} (c1)
	(c0) edge[below, bend left=15] node{$\csend{c}{s}{buy}$} (c2)
	(c1) edge[right, bend left, pos=0.25] node{$\crec{s}{c}{price}$} (c0)
	(c2) edge[right, bend left=15] node{$\crec{b}{c}{ask\_pay}$} (c3)
	(c3) edge[below, bend left=15] node{$\csend{c}{b}{info\_pay}$} (c4)
	(c4) edge[left, bend left=15] node{$\crec{s}{c}{cancel}$} (c0)
	(c4) edge[right, bend right] node{$\crec{s}{c}{conf}$} (c0)
	
	(s0) edge[left, bend left] node{$\crec{c}{s}{ask\_price}$} (s1)
	(s0) edge[above] node{$\crec{c}{s}{buy}$} (s2)
	(s1) edge[right, bend left, pos=0.25] node{$\csend{s}{c}{price}$} (s0)
	(s2) edge[above] node{$\csend{s}{b}{transaction}$} (s3)
	(s3) edge[below right, bend left] node{$\crec{b}{s}{auth}$} (s4)
	(s3) edge[above right, bend right] node{$\crec{b}{s}{refusal}$} (s5)
	(s4) edge[below left, bend left] node{$\csend{s}{c}{conf}$} (s0)
	(s5) edge[below right, bend right] node{$\csend{s}{c}{cancel}$} (s0)
	
	(b0) edge[above] node{$\crec{s}{b}{transaction}$} (b1)
	(b1) edge[above] node{$\csend{b}{c}{ask\_pay}$} (b2)
	(b2) edge[above] node{$\crec{c}{b}{info\_pay}$} (b3)
	(b3) edge[above, bend right=30] node{$\csend{b}{s}{refusal}$} (b0)
	(b3) edge[below, bend left=20] node{$\csend{b}{s}{auth}$} (b0);
	\end{tikzpicture}
	\caption{Client/Seller/Bank}\label{fig:example:csb}
\end{figure}
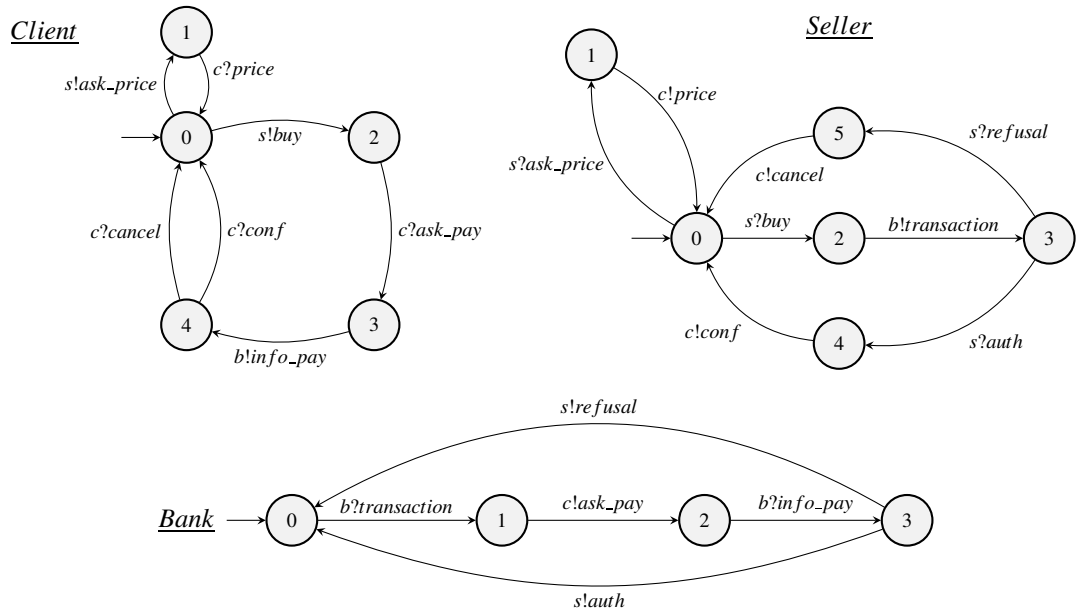

\begin{example}
To finish this section we  present another small example of a 
\halfduplex system: the classic Client/Seller/Bank protocol.
This system is shown in Figure \ref{fig:example:csb}. 
In this protocol a client can ask the price of an item to the 
seller ($ask\_price$ message), and receive the answer. Whenever 
the client agrees on a price it can place an order (via the 
message $buy$). Receiving this message the seller initiates a 
transaction with the bank. The bank asks the client for its 
credentials, and when it receives them it either authorizes or 
refuses the transaction, and
notifies the seller accordingly. The
seller then confirms or cancels the transaction, sending a 
message to the client.
\qed
\end{example}

\section{Conclusion}\label{sec:conclusion}
% !TEX root = ice.tex

We have introduced greedy systems, a new class of communicating 
systems, generalising  the notion of \halfduplex systems
to any number of processes, and to an arbitrary model of FIFO 
communication (encompassing both p2p and mailbox communications).
We have shown that the greediness of a system is decidable
in polynomial time, and that for greedy systems regular
safety properties, such as 
reachability, progress, and boundedness are decidable
in polynomial time. Finally, we defined
mailbox half-duplex systems and showed that greedy
systems are intimately related to mailbox half-duplex systems.

Still, the picture is a bit incomplete: we did not
address liveness properties nor more general temporal
properties, and we also did not propose a notion of p2p half-duplex
systems that would enjoy all desirable properties. Also, we did not
report on experimental evaluation. We leave
these questions for future work.

% Notice that the decidability of boundedness can also be obtained 
% directly with an  algorithm  
% that is reminiscent of the Karp-Miller coverability algorithm 
% for Petri nets. The idea is that in a \halfduplex system the 
% size of a buffer can grow unboundedly only because of an 
% arbitrary number of  unmatched sends. Results in \cite{DBLP:journals/sosym/FinkelL15} can thus be used to compute 
% a finite basis for the coverability set.

Pachl gave a general decidability result for systems of communicating finite state machines whose reachability
set is regular~\cite{DBLP:journals/corr/cs-LO-0306121}.
Although the algorithm is rather brute-force, the MCSCM tool
illustrates that it is amenable to an efficient CEGAR optimization~\cite{DBLP:conf/tacas/HeussnerGS12}.
Beyond regularity, and relying on visibly pushdown languages, La Torre \emph{et al}~\cite{DBLP:conf/tacas/TorreMP08} established that bounded context-switch reachability is also decidable
in time exponential in the number of states and doubly exponential in the number of switch. We conjecture that bounded context-switch reachability is complete for greedy systems with a number of switch polynomial in the size of the system.
 
Several authors considered communicating systems where queues are not plain FIFO but may over-approximate a FIFO behaviour. 
A representative example of this approach is lossy channels, introduced independently in \cite{DBLP:journals/iandc/CeceFI96} 
and \cite{DBLP:journals/iandc/AbdullaJ96}.
Another example is the bag semantics of buffers where messages 
are received without loss but out of order. Examples of uses of 
bag buffers can be found in \cite{Koushik}.
Both for lossy and bag systems the reachability problem is 
decidable but with a high complexity: 
non primitive recursive for lossy~\cite{DBLP:journals/ipl/Schnoebelen02, DBLP:conf/csl/Schnoebelen21}. The exact 
complexity seems 
unknown for bag systems, but by reduction to Petri nets
it is non-elementary~\cite{DBLP:journals/jacm/CzerwinskiLLLM21}.

Aside bounded context-switch model-checking, another form of bounded model-checking has been promoted by Muscholl \emph{et al.}: the notions of universally and existentially bounded message
sequence charts~\cite{DBLP:journals/iandc/LohreyM04}. This
leads to two notions of universally/existentially bounded 
systems (both having unfortunately the same name), depending 
whether only the MSCs leading to final stable configurations are 
considered~\cite{DBLP:journals/fuin/GenestKM07} or all 
MSCs~\cite{KuskeM2018}. For the latter definition, reachability 
and membership are decidable in PSPACE. Greedy systems are
existentially $1$-bounded, but the converse does not hold
(for instance, the system of Example~\ref{ex:counter-example-greedy})
is existentially $1$-bounded).
A system is 
$k$-synchronous~\cite{DBLP:conf/cav/BouajjaniEJQ18} if the 
message sequence charts
of the system are formed of blocks of $k$ messages.
In particular, a system is $1$-synchronous if the message lines never cross.
For systems with p2p communications, greediness is the same
as $1$-synchronizability. However, for mailbox communications,
some subtle examples are 
$1$-synchronous but not greedy (see 
\emph{e.g.}~\cite[Example~1.2]{DBLP:conf/fossacs/GiustoLL20}). 
For $k$-synchronous systems, reachability is 
decidable in PSPACE. Finally, greedy system are synchronizable
in the sense of Basu and Bultan~\cite{DBLP:journals/tcs/BasuB16}, 
but synchronizability is not decidable~\cite{FL-icalp17}. We 
believe that greediness is the notion that Basu and Bultan were
aiming at with the notion of synchronizability.
It might be wondered if greedy systems were not implicit in
Cece and Finkel's work. Actually, some of their arguments
rely on the fact that for half-duplex systems, every execution
is \textquote{reachability equivalent} to a synchronous execution.
This is not exactly the notion of greedy systems we introduced,
and our notion of greedy systems is closer to Bouajjani~et~al notion 
of $1$-synchronous systems, although, as we just explained, they
are not the same in some corner cases.

Session types are intimately related to half-duplex systems in
the binary setting~\cite{DBLP:conf/wsfm/LozesV11}. Several
multi-party extensions of session types have been proposed,
the last proposal being~\cite{DBLP:journals/pacmpl/ScalasY19}.
It seems there are also some similarities between multiparty session
types and greedy systems. For instance, the notion of greedy
execution shares some similarities with the notion of
alternation in~\cite{DBLP:conf/esop/DenielouY12}.
The study of the exact relationship between greedy systems and 
multi-party session typed systems is left for future work. 

\paragraph{Acknowledgements.} We would like to thank all the ICE reviewers for their comments
that greatly improved the present paper.
%and as our results rely heavily on the 1-synchronizability.

\bibliographystyle{eptcs}
 \bibliography{biblio}

\end{document}